\documentclass{article}

\usepackage{amssymb}
\usepackage{amsthm}
\usepackage{color}
\usepackage{comment}
\usepackage{fullpage}
\usepackage{graphicx}
\usepackage{stmaryrd}
\usepackage{courier}
\usepackage{authblk}

\newtheorem{theorem}{Theorem}[section]
\newtheorem{definition}[theorem]{Definition}

\newtheorem{proposition}[theorem]{Proposition}
\newtheorem{lemma}[theorem]{Lemma}
\newtheorem{corollary}[theorem]{Corollary}

\newtheorem{remark}[theorem]{Remark}
\newtheorem{fact}[theorem]{Fact}

\def\Supp{\textnormal{Supp }}

\def\Im{\textnormal{Im }}

\begin{document}

\title{The Impossibility of Efficient Quantum Weak Coin-Flipping}

\author{Carl A.~Miller}

\affil{\small Joint Center for Quantum Information and Computer Science \\  University of Maryland, College Park, MD 20742, USA} \affil{\small National Institute of Standards and Technology, \\ 100 Bureau Dr., Gaithersburg, MD 20899, USA}

\date{}

\maketitle

\abstract{How can two parties with competing interests carry out a fair coin flip, using only a  noiseless quantum channel? This problem (quantum weak coin-flipping) was formalized more than 15 years ago, and, despite some phenomenal theoretical progress, practical quantum coin-flipping protocols with vanishing bias have proved hard to find.  
In the current work we show that there is a reason that practical weak quantum coin-flipping is difficult: any quantum weak coin-flipping protocol with bias $\epsilon$ must use at least $\exp ( \Omega (1/\sqrt{\epsilon} ))$ rounds of communication.  This is a large improvement over the previous best known lower bound of $\Omega ( \log \log (1/\epsilon ))$ due to Ambainis from 2004.  Our proof is based on a theoretical construction (the two-variable profile function) which may find further applications.
}

%\tableofcontents

\section{Introduction}

Suppose that Alice and Bob are two cooperating but mutually mistrustful parties, and they must  make a unified decision between two choices ($X$ and $Y$).  Alice wants choice $X$, and Bob wants choice $Y$.  However, neither of them will gain if they do not agree on their decision.
How can the decision be made fairly?
A natural solution would be for Alice and Bob to have a trusted third party (Charlie) flip a coin and report the result to both Alice and Bob.  But, can coin-flipping be done in absence of any trusted third party or common source of randomness?  

In this paper we will be concerned with the question of whether coin-flipping can be done if Alice and Bob share a two-way noiseless quantum channel.  A standard way to model a protocol in this scenario is like so (see Figure~\ref{fig:protocol}).  Let $n$ be an even positive integer.
\begin{enumerate}
    \item \label{firstassump} Alice possesses a quantum system $\mathcal{A}$, which she controls, and Bob possesses a quantum system $\mathcal{B}$ which he controls.
    
    \item There is an additional quantum system $\mathcal{M}$, initially possessed by Alice, which stores quantum messages exchanged by Alice and Bob during the protocol.
    
    \item If $i$ is odd, then on the $i$th round of the protocol, Alice performs a prescribed joint quantum operation on $\mathcal{A}$ and $\mathcal{M}$, and then sends $\mathcal{M}$ across the quantum channel to Bob.
    
    \item If $i$ is even, then on the $i$th round of the protocol, Bob performs a prescribed joint quantum operation on $\mathcal{B}$ and $\mathcal{M}$ and then sends $\mathcal{M}$ across the channel to Alice.
    
    \item \label{lastassump} After the $n$th round of communication, Alice performs a binary measurement on $\mathcal{A}$ and reports the result as a bit ($a$), and Bob performs a binary measurement on $\mathcal{B}$ and reports the result as a bit ($b$).
\end{enumerate}
\begin{figure}
    \centering
    \includegraphics[scale=0.25]{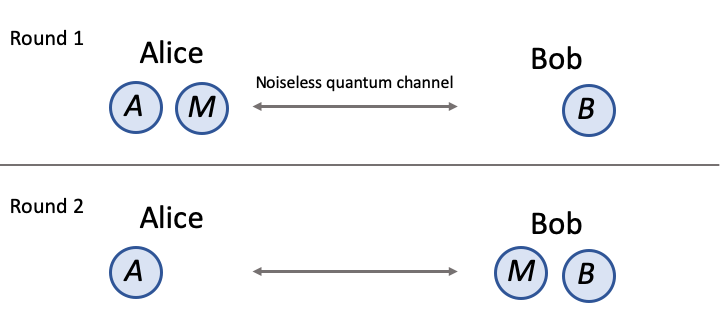}
    \caption{The first two rounds of a weak quantum coin-flipping protocol.}
    \label{fig:protocol}
\end{figure}
We suppose that Alice desires the outcome $a = b = 0$, and that Bob desires the outcome $a=b=1$.
It is presumed that an ``honest'' party will carry out their operations and measurements exactly as prescribed; however, a dishonest party may perform arbitrary manipulations of the quantum systems that they possess, and may perform any final measurement that they choose at the end of the protocol.  (In particular, there is no restriction on the computational power of Alice or Bob.)

We say such a protocol is a \textbf{weak coin-flipping protocol with bias $\epsilon$} if the following hold:
\begin{enumerate}
    \item If Alice and Bob both perform honestly, then $\mathbb{P} ( a = b = 0)$ is exactly $\frac{1}{2}$ and $\mathbb{P} ( a = b = 1 )$ is exactly $\frac{1}{2}$.
    
    \item If Alice behaves dishonestly and Bob behaves honestly, then $\mathbb{P} ( b = 0) \leq \frac{1}{2} + \epsilon$.
    
    \item If Bob behaves dishonestly and Alice behaves honestly, then $\mathbb{P} ( a = 1 ) \leq \frac{1}{2} + \epsilon$.
\end{enumerate}
These conditions assert that Alice cannot bias the outcome by more than $\epsilon$ in her favor, and Bob cannot bias the result by more than $\epsilon$ in his favor.
(A \textbf{strong coin-flipping protocol with bias $\epsilon$} is one which guarantees that a dishonest party cannot bias the result of the coin flip by more than $\epsilon$ in either direction.  Strong coin-flipping with vanishing bias is impossible by an elementary argument attributed to A.~Kitaev --- see  \cite{Landau:2014}.)

We note that in a classical analogue of the setting described above --- i.e., where $\mathcal{A, B, M}$ are random variables, Alice's and Bob's operations are stochastic maps, and the two parties are computationally unlimited --- it is elementary to show that any protocol allows either Alice to force the outcome to always be $0$, or Bob to force the outcome to always be $1$. (Thus, achieving bias less than $\frac{1}{2}$ is impossible.) This fact can be overcome by putting computational restrictions on Alice and Bob, and classical coin-flipping is itself an elegant and extensively studied topic (see \cite{Blum:1983, Moran:2016}).   The main motivation to study quantum coin-flipping is that, in contrast to classical coin-flipping, it allows security proofs based on physical assumptions only.

Quantum coin-flipping was formalized as early as 1998 (\cite{Lo:1998}), and a series of works proved coin-flipping protocols with progressive improvements in the bias. Aharonov et al.~\cite{Aharonov:2000} proved bias $0.42$. Spekkens and Rudolph, and independently Ambainis, proved successive results \cite{Spekkens:2001,Spekkens:2002, Ambainis:2004} which brought the bias down to $\frac{\sqrt{2}-1}{2} \approx 0.207$.  These results involved a small constant number of rounds of communication.  Mochon \cite{Mochon:2004,Mochon:2005} then introduced a family of quantum weak coin-flipping protocols which approach bias $1/6 \approx 0.166$, with the number of communication rounds tending to infinity.

Finally, in a landmark work in 2007, Mochon \cite{Mochon:2007} showed the existence of a family of weak coin-flipping protocols with bias tending to zero.  
Mochon exploited the idea of \textit{point games} (a concept also attributed to A.~Kitaev) to achieve this result.
Mochon's existence proof  was later simplified, re-written and published by Aharonov et al.~\cite{Aharonov:2016}.
Then, in recent work \cite{Arora:2019}, Arora et al.~introduced an algorithm which effectively constructs the protocols in the family whose existence was proven by Mochon.  

Following this phenomenal progress, at least one major loose end remains.  The number of communication rounds used in the protocols in \cite{Mochon:2007, Aharonov:2016} was only shown to be $ (1/\epsilon)^{O ( 1/\epsilon )}$.  This asymptotic quantity is hardly efficient or practical.  Meanwhile, the best known \textit{lower} bound on the number of communication rounds \cite{Ambainis:2004} is $\Omega ( \log \log (1 / \epsilon ))$, leaving a vast range of uncertainty about the optimal resources needed to achieve vanishing bias.   

How many rounds of quantum communication are needed to achieve a particular bias $\epsilon$?  For example, in 
a different setting (strong classical coin-flipping with computational hardness assumptions), Cleve \cite{Cleve:1986} showed that the number of communication rounds is $\Omega ( 1 / \epsilon )$, and this bound was shown to be achievable \cite{Moran:2016}.
Could a similar relationship exist for quantum weak coin-flipping?

\subsection{Summary of result}

In the current paper, we prove the following lower bound on the number of communication rounds for quantum weak coin-flipping (see Theorem~\ref{thm:final}):
\begin{theorem}
\label{thm:finalquote}
Let $\mathbf{C}$ be an $n$-round quantum weak coin-flipping protocol with bias $\epsilon$.  Then,
\begin{eqnarray}
n & \geq & \exp \left( \Omega \left( \frac{1}{\sqrt{\epsilon}} \right) \right).
\end{eqnarray}
\end{theorem}
This result shows that, at least in the standard model, practical quantum weak coin-flipping with vanishing bias is not feasible. 

The proof of this result builds on previous techniques, including the concept of a \textit{valid time-independent point game} (abbreviated as ``valid TIPG'').  A valid TIPG is a pair of real-valued functions $m_1, m_2$ on $\mathbb{R}_{\geq 0} \times \mathbb{R}_{\geq 0}$ that satisfy a certain infinite set of linear constraints (see subsections~\ref{subsec:vpg}--\ref{subsec:relationship}).  
It is known that any weak coin-flipping protocol determines a valid TIPG, and vice versa.  This correspondence was used to prove the family of protocols with vanishing bias in \cite{Mochon:2007}.
Here I prove a negative result: any TIPG obtained from a weak coin-flipping protocol with small bias must have very large $1$-norm --- i.e., $\left\| m_1 \right\|_1 + \left\| m_2 \right\|_1$ must be very large as a function of $\epsilon$.  Since there is a relationship between the number of communication rounds of a protocol and the $1$-norms of its associated time-independent point games, this implies the main result.

Ambainis's original bound \cite{Ambainis:2004} of $\Omega ( \log \log (1/\epsilon ) )$ was based on a more direct study of quantum weak coin-flipping protocols: he performed an inductive argument, using the fidelity function, on the intermediate states arising in the protocol.  This approach appears to be fairly different from the one in this paper, although it may be possible to relate the two.

Theorem~\ref{thm:finalquote} is a further step in mapping out  the full range of cryptographic possibilities  in a two-party quantum setting (see \cite{Broadbent:2016, Wehner:2018} for surveys on this topic).  A number of other negative results are known: secure two-party computation (under certain  definitions) is impossible \cite{Lo:1997, Buhrman:2012}, and  strong coin-flipping \cite{Landau:2014, Chailloux:2017}, bit commitment~\cite{Lo:1997b,Mayers:1997, Chailloux:2017}, and oblivious transfer~\cite{Lo:1997, Chailloux:2013} are all impossible except with fixed positive bias.
This paper shows that the case of quantum weak coin-flipping is different: it can be achieved with arbitrarily small bias, but it is impossible to do so in polynomial time.

Certainly, this is not the end of the story.  The model for quantum weak coin flipping makes a number of assumptions, including  that the players exchange information in discrete stages, and that they are completely unconstrained in their ability to manipulate any quantum systems that are not under the control of the other player. This impossibility result gives us additional motivation to study coin-flipping in other settings, including relativistic models.

\subsection{Outline} 

Sections \ref{sec:prelim}--\ref{sec:coin} of this paper cover preliminaries and known material, and then new contributions appear in sections~\ref{sec:profile}--\ref{sec:lp4}.
In section~\ref{sec:profile}, I define the  \textit{profile} of a time-independent point game, which is a two-variable function associated to a time-independent point game that distills some of its most relevant information.  
Sections \ref{sec:highly}--\ref{sec:expect} prove some mathematical lemmas, including a key result on the behavior of highly concentrated rational functions (Proposition~\ref{prop:complexlarge}).  Section~\ref{sec:commbound} proves the result about the  $1$-norm of a time-independent point game using tools from the previous sections.
Finally, section~\ref{sec:lp4} proves Theorem~\ref{thm:finalquote}.   I conclude by noting some further directions in section~\ref{sec:dir}.

\subsection{Acknowledgements}

I am grateful to Aarthi Sundaram (my co-author on a companion project) for many interesting discussions about the coin-flipping literature which helped to seed some of the ideas contained here.  This paper also owes a large debt to Alexandre Eremenko, who showed me the complex analysis method that I used to prove Proposition~\ref{prop:complexlarge}.  
Thanks to Gorjan Alagic, Jonathan Katz, Yi-Kai Liu, and Scott Wolpert for their help with this project, and to Michael Newman for giving me an introduction to \cite{Aharonov:2016} some years ago.

This work is a contribution of the U.~S.~National Institute of Standards and Technology (NIST), and is not subject to copyright in the United States.  Any mention of commercial products is for information purposes only, and does not imply endorsement by NIST.

%\section{Proof Overview}

%\subsection{Review of point games}

%\subsection{The profile functions}

%\subsection{Concentration results}

%\label{subsec:pointgames}

\section{Preliminaries}

\label{sec:prelim}

%\CM{Note that, although we refer heavily to the work of Aharonov et al., the proof that we are building on is really Mochon's.}

Let $\mathbb{R}$ denote the set of real numbers, and let $\mathbb{R}_{\geq 0}$ denote the set of nonnegative real numbers. 
If $a,b$ are real numbers with $a \leq b$, then $[a,b]$ denotes the closed interval $\left\{ x \mid a \leq x \leq b \right\}$, $(a,b)$ denotes the open interval $\left\{ x \mid a < x < b \right\}$, and $[a,b)$ and $(a,b]$ are similarly defined.  We let $\infty$ denote infinity, and define intervals such as $[a,\infty] \subseteq \mathbb{R} \cup \left\{ \infty \right\}$ in the obvious way.  If $A$ is a set and $B \subseteq A$ is a subset, then $A \smallsetminus B$ denotes
the set of all elements of $A$ that are not in $B$.

Our notation follows previous work \cite{Mochon:2007, Aharonov:2016, Arora:2019} in part.
Since we will work extensively with functions on $\mathbb{R}_{\geq 0}$ that have finite support, we make the following definitions.
\begin{definition}
For any $x \in \mathbb{R}_{\geq 0}$, let $\llbracket x \rrbracket$ denote the function from $\mathbb{R}_{\geq 0}$ to $\mathbb{R}$ which maps $x$ to $1$ and is zero elsewhere.
For any $x, y \in \mathbb{R}_{\geq 0}$, let $\llbracket x,y \rrbracket $ denote the function from $\mathbb{R}_{\geq 0} \times \mathbb{R}_{\geq 0}$ to $\mathbb{R}$ which maps $(x,y)$ to $1$ and is zero elsewhere.
\end{definition}

If $f$ is a real-valued function on a set $S$, define $f^+ \colon S \to \mathbb{R}$ and $f^- \colon S \to \mathbb{R}$ by
\begin{eqnarray}
f^+ (x ) & = & \textnormal{max} \{ 0, f (x) \}, \\
f^- (x ) & = & \textnormal{max} \{ 0, -f (x) \}.
\end{eqnarray}
Note that $f = f_+ - f_-$.  Let $\Supp f$ denote the support of $f$ (i.e., the set of points in $S$ on which $f$ is nonzero). Let $f^\top$ denote the function $f^\top (x, y) = f(y, x)$.  When $f$ is a function with finite support, then (even if the set $S$ is not countable) we will write $\sum_{s \in S} f ( s )$ to mean the sum of $f ( s)$ over all points in the support of $f$.  The expression $\left\| f \right\|_1$ denotes the sum $\sum_{s \in S} \left| f ( s ) \right|$ (that is, the $1$-norm of $f$).

The function $\log \colon \mathbb{R}_{\geq 0} \to \mathbb{R} \cup \{ - \infty \}$ denotes the logarithm in base $2$.  

We use the term \textbf{universal function} to mean a function that is not dependent on any variables other than its input variables.  Thus, even if we refer to a universal function after some variables have been quantified (e.g., ``for all $c$, ...'') it is understood that the function has no implicit dependencies on those variables.  We will use boldface Roman letters ($\mathbf{A}, \mathbf{B}, \ldots$) for universal functions.

When we use asymptotic big-$O$ notation, we may use $O( u )$ as a set (e.g., ``there exists $\mathbf{F}(u) \in O ( u )$ such that ...'') or as a placeholder for a function (e.g., ``$x = y + O ( z )$'').
When a big-$O$ expression is used as a placeholder, it is understood that it also represents a universal function with no implicit dependencies. 

If $Q$ is an event, then we write $\mathbb{P} [ Q ]$ for the probability of $Q$, and if $X$ is a real-valued random variable, then we write $\mathbb{E} [ X ]$ for the expectation of $X$. A \textbf{stochastic map} from a  set $A$ to a set $B$ is an indexed set of nonnegative real values
\begin{eqnarray}
V & = & \left\{ v_{ab} \mid a \in A, b \in B \right\}
\end{eqnarray}
such that $\sum_b v_{ab} = 1$ for all $a \in A$.  For any $a \in A$, the values $\left\{ v_{ab} \mid b \in B \right\}$ define a random variable on $B$ which we denote by $V ( a )$.  

If $\mathcal{A}$ and $\mathcal{B}$ are Hilbert spaces, then we may write $\mathcal{A} \mathcal{B}$ for the tensor product $\mathcal{A} \otimes \mathcal{B}$.

\subsection{Complex analysis}

\label{subsec:complex}

We briefly cover some complex analysis tools that will be important in section~\ref{sec:highly}.  The reader can consult \cite{Ahlfors:1979} for more details.

Let $\mathbb{C}$ denote the set of complex numbers.  We will apply addition and multiplication to $\mathbb{C} \cup \left\{ \infty \right\}$ using natural rules ($c + \infty = \infty, 1/\infty = 0$, etc.).  For any $z \in \mathbb{C}$ and $r \geq 0$, let
\begin{eqnarray}
\mathbb{D} ( z, r ) & = & \left\{ w \in \mathbb{C} \mid \left| z - w \right| < r \right\} \\
\mathbb{S} ( z, r ) & = & \left\{ w \in \mathbb{C} \mid \left| z - w \right| = r \right\}.
\end{eqnarray}
(The set $\mathbb{D} ( z, r )$ is the open disc of radius $r$ centered at $z$, and the set $\mathbb{S} ( z, r )$ is the circle of radius $r$ centered at $z$.)
When $z = 0$ and $r = 1$, we may write these sets simply as $\mathbb{D}$ and $\mathbb{S}$.  Also let
\begin{eqnarray}
\mathbb{H}  & = & \left\{ w \in \mathbb{C} \mid \Im w > 0 \right\}.
\end{eqnarray}
If $Y \subseteq \mathbb{C} \cup \{ \infty \}$, then $\overline{Y}$ denotes the closure of $Y$.
For ease of notation we will write $\overline{\mathbb{D}} ( z, r )$ for the closure of $\mathbb{D} ( z, r )$.  Let $\overline{\mathbb{H}}$ denote the set $\left\{ z \in \mathbb{C} \mid \Im z \geq 0 \right\} \cup \{ \infty \}$.

If $S \subseteq \mathbb{C}$ is an open set, then a function $f \colon S \to \mathbb{C}$ is \textbf{analytic} if for any $p \in S$, $f$ can be expressed as a power series on some open neighborhood of $p$.  If $f$ is analytic and $\overline{\mathbb{D}} ( z, r )$ is a closed disc within its domain, then
the following equation always holds:
\begin{eqnarray}
\frac{1}{2 \pi} \int_{0}^{2\pi} f ( z + r e^{i \theta} ) d \theta & = & f ( z ).
\end{eqnarray}

\section{Review of quantum weak coin-flipping}

\label{sec:coin}

In this section we review the common formal framework 
for quantum weak coin-flipping, including the mathematical construction of a \textit{point game} (which is attributed to A.~Kitaev).  Since this framework has already seen thorough treatment in
\cite{Mochon:2007, Aharonov:2016, Arora:2019}, 
we will mainly provide only definitions and statements of results here. Our terminology and notation are derived most directly from
\cite{Aharonov:2016}.

\subsection{Weak coin flipping protocols}

\label{subsec:coinintro}

The definition of a weak coin flipping protocol (which we will first sketch, and then state formally) is intended to capture a general situation where two parties with competing interests are trying to fairly flip a coin.    There are two possible outcomes: $0$ (or ``heads'') which is the desired outcome for Alice, and $1$ (or ``tails'') which is the desired outcome for Bob.  There is no trusted third party in the protocol, and thus it consists entirely of communication between Alice and Bob.  At the end of the protocol, the parties report bits $a$ and $b$ respectively (representing what they ostensibly believe to be the outcome of the coin flip).

The protocol is accomplished by Alice and Bob passing a quantum system (represented by the finite-dimensional Hilbert space $\mathcal{M}$) back and forth  between them, while keeping private systems (represented by $\mathcal{A}$ and $\mathcal{B}$) to themselves.  In each odd round $i$, Alice performs a unitary operator $U_i$ on $\mathcal{A} \mathcal{M}$ followed by a binary projective measurement $\left\{ E_i, I_{\mathcal{AM}} - E_i \right\}$ on $\mathcal{AM}$.  If the latter measurement fails --- that is, if its postmeasurement state is not in $\Supp E_i$ --- then Alice aborts the protocol and simply reports her favored outcome $0$.   
(This event is understood to mean that Alice has stopped because she suspects cheating.)
On even rounds, Bob does analogous operations with operators $U_i, E_i$ on $\mathcal{MB}$.  At the end of the protocol, if neither party has aborted, they each perform a binary projective measurement on their private system and report the result (as $a$ and $b$, respectively).  The definition requires that if Alice and Bob behave honestly in the protocol, then the probability that $a = b = 0$ is $1/2$, and the probability that $a = b = 1$ is $1/2.$

\begin{definition}
\label{def:coinprotocol}
A \textbf{weak coin-flipping protocol} $\mathbf{C}$ consists of the following data:
\begin{itemize}
    \item Finite-dimensional Hilbert spaces $\mathcal{A}, \mathcal{M}, \mathcal{B}$ (Alice's system, the message system, and Bob's system),
    
    \item An even positive integer $n$ (the number of rounds),
    
    \item An initial pure state $\psi_0$ on $\mathcal{AMB}$ of the form
    \begin{eqnarray}
    \psi_0 & = & \psi_{\mathcal{A}, 0} \otimes
    \psi_{\mathcal{M}, 0} \otimes \psi_{\mathcal{B}, 0},
    \end{eqnarray}

    \item For each odd value $i$ from $\{ 1, 2, \ldots, n \}$, a unitary operator $U_i$ on $\mathcal{AM}$ and a Hermitian projection operator $E_i$ on $\mathcal{AM}$,

    \item For each even value $i$ from $\{ 1, 2, \ldots, n \}$, a unitary operator $U_i$ on $\mathcal{MB}$ and a Hermitian projection operator $E_i$ on $\mathcal{MB}$,

    \item Binary projective measurements $\{ \Pi_{\mathcal{A}}^0, \Pi_{\mathcal{A}}^1 \}$ and $\{ \Pi_{\mathcal{B}}^0, \Pi_{\mathcal{B}}^1 \}$ on $\mathcal{A}$ and $\mathcal{B}$, respectively,
\end{itemize}
where the following condition holds: the state
\begin{eqnarray}
\psi_n & := & E_n U_n E_{n-1} U_{n-1} \cdots E_1 U_1 \psi_0
\end{eqnarray}
(which is referred to as the \textbf{final state} of the protocol) satisfies
\begin{eqnarray}
&
\left\|  \Pi_{\mathcal{A}}^0 \otimes  \Pi_{\mathcal{B}}^0 \left| \psi_n \right> \right\|^2  =    \left\| \Pi_{\mathcal{A}}^1 \otimes \Pi_{\mathcal{B}}^1 \left| \psi_n \right> \right\|^2   =  \frac{1}{2}.
\end{eqnarray}
\end{definition}
For the definition above, the states 
\begin{eqnarray}
\psi_i & := & E_i U_i E_{i-1} U_{i-1} \cdots E_1 U_1 \psi_0
\end{eqnarray}
for $i \in \{ 1, \ldots, n-1 \}$, are referred to as the \textbf{intermediate states} of the protocol.

Let us suppose that Bob (whose goal is to force Alice to report $a = 1$) chooses to behave dishonestly in the protocol.  In that case, he can apply arbitrary unitary operations $V_2 , V_4, V_6, \ldots$ on $\mathcal{M} \mathcal{B}$ in place of $E_2 U_2, E_4 U_4, E_6 U_6, \ldots$.  (We do not account for any measurements performed by Bob in this case, since his own output is irrelevant.)  
This motivates the following definition.    The \textbf{cheating probability} for Bob (in protocol $\mathbf{C}$) is the maximum of 
\begin{eqnarray}
\left\| \Pi_\mathcal{A}^1 V_n (E_{n-1} U_{n-1}) V_{n-2} (E_{n-3} U_{n-3} ) V_{n-4} (E_{n-5} U_{n-5} ) \ldots V_2 (E_1 U_1 ) \left| \psi_0 \right> \right\|^2
\end{eqnarray}
over all unitary operators $V_2, V_4, \ldots, V_n$ on $\mathcal{MB}$.  

Let $P_B^*$ denote the cheating probability for Bob, and let $P_A^*$ denote the cheating probability for Alice (defined analogously).
Then, the $\textbf{bias}$ of the weak coin-flipping protocol $\mathbf{C}$ is the quantity
\begin{eqnarray}
\textnormal{max} \left\{ P_A^* - \frac{1}{2} , P_B^*
- \frac{1}{2} \right\}.
\end{eqnarray}

\subsection{Valid point games}

\label{subsec:vpg}

Valid point games are elegant mathematical constructions which, as we will see, are in a near-perfect correspondence with weak coin-flipping protocols.  Because they are a lot simpler to define, valid point games provide a convenient method of reduction for questions about weak coin-flipping protocols.  We will give the definition of valid point games in this subection, and then explain their relationship to weak coin-flipping protocols in subsection~\ref{subsec:relationship}.
We use standard terminology with a few additions.

\begin{definition}
A \textbf{one-dimensional move} is a function $\ell \colon \mathbb{R}_{\geq 0} \to \mathbb{R}$ that has finite support.  A one-dimensional move is called a one-dimensional \textbf{configuration} if $\ell \geq 0$.
A \textbf{two-dimensional move} is a function $q \colon \mathbb{R}_{\geq 0} \times \mathbb{R}_{\geq 0} \to \mathbb{R}$ that has finite support.  A two-dimensional move is called a two-dimensional \textbf{configuration} if $q \geq 0$.
\end{definition}

\begin{remark}
When we use the words \textbf{move} or \textbf{configuration} by themselves, we will always mean a two-dimensional move or two-dimensional configuration.
\end{remark}

Given a move $q$, it can be helpful to visualize $q$ by graphing its support set $\Supp q$ and writing out the values $q (x, y)$ associated to each point $(x,y) \in \Supp q$. See Figure~\ref{fig:move} for an example of a move whose support is of size $5$.
\begin{figure}
    \centering
    \includegraphics[scale=0.5]{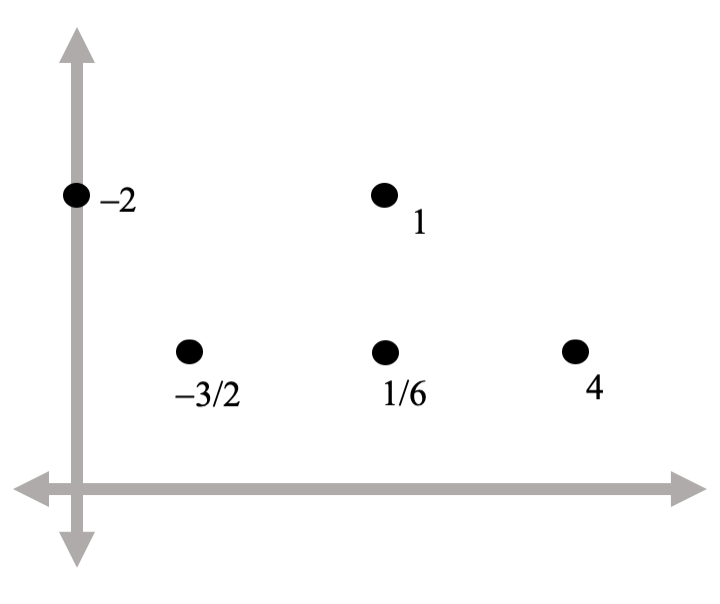}
    \caption{An example of a diagram of a $2$-dimensional move $w$ on a Cartesian coordinate system.  The marked points are the elements of $\Supp w$, and the labels indicate the values taken by $w$.}
    \label{fig:move}
\end{figure}

A time-dependent point game is, roughly speaking, a finite sequence of two-dimensional configurations
$z_0, \ldots, z_n$.  However, for mathematical convenience, we define time-dependent point games in terms of the moves $(z_i - z_{i-1})$ rather than the configurations $z_i$.  Recall that for any real-valued function $f$, we write $f^+$ and $f^-$ for the positive and negative parts of $f$, respectively.

\begin{definition}
\label{def:tdpg}
A time-dependent point game $M$ is a sequence of moves $(m_1, \ldots, m_n)$ such that
\begin{eqnarray}
\left( \sum_{i=1}^n m_i \right)^- + \sum_{i=1}^j m_i & \geq & 0
 \label{impositivity}
\end{eqnarray}
for all $j \in \{ 1, 2, \ldots, n-1 \}$.
\end{definition}
In the above definition, we refer to  $\left( \sum_{i=1}^n m_i \right)^-$ as the \textbf{initial configuration} of $M$, and we refer to $\left( \sum_{i=1}^n m_i \right)^+$ as the \textbf{final configuration} of $M$.  The configurations on the left side of inequality (\ref{impositivity}) are referred to as the \textbf{intermediate configurations} of $M$.  
The \textbf{support} of a point game $M$, denoted $\Supp M$, is the union of the supports of the moves in $M$. 

Next we will define valid moves. 
\begin{definition}
\label{def:valid}
A one-dimensional move $\ell$ is \textbf{valid} if 
the following conditions hold:
\begin{eqnarray}
\sum_x \ell (x ) & = & 0 \\
\label{valid2}
\sum_x \left( \frac{x  }{x + \lambda} \right) \ell ( x ) & \geq & 0  
\hskip0.2in \forall \lambda > 0 \\
\label{valid3}
\sum_x x \cdot \ell (x ) & \geq & 0.
\end{eqnarray}
\end{definition}
(We note that condition (\ref{valid3}) is actually redundant, since it can be proved from (\ref{valid2}).) Equivalently, a move is valid if and only if its inner product with any operator monotone function is nonnegative (see subsection 3.2 of \cite{Aharonov:2016}).
We note the following useful fact, which is easily proved from Definition~\ref{def:valid}.
\begin{proposition}
\label{prop:stretch}
If $\ell$ is a valid one-dimensional move, and $c > 0$, then the function $x \mapsto \ell ( cx)$ is also a valid one-dimensional move.  $\Box$
\end{proposition}

For any two-dimensional move $q$, the \textbf{rows} of $q$ are the functions of the form $x \mapsto q ( x, y )$ (for $y \in \mathbb{R}_{\geq 0}$) and the \textbf{columns} of $q$ are the functions of the form $y \mapsto q ( x, y)$ (for $x \in \mathbb{R}_{\geq 0}$).

\begin{definition}
A two-dimensional move is \textbf{horizontally valid} if its rows are valid.  A two-dimensional move is \textbf{vertically valid} if its columns are valid.  A point game $(m_1, \ldots, m_n)$ is a \textbf{valid point game} if $m_i$ is horizontally valid for all odd $i$, and $m_i$ is vertically valid for all even $i$.
\end{definition}

For later use, we define the concept of a time-independent point game, which is simply a sequence of  two-dimensional moves with no restrictions on nonnegativity (unlike Definition~\ref{def:tdpg}).  
In this paper, as in previous papers on weak coin-flipping, it is only useful to consider time-independent point games that involve $2$ moves, and so we confine our definition accordingly. 

\begin{definition}
A \textbf{time-independent point game (TIPG)} is a pair of moves $(r_1, r_2)$.  
\end{definition}

We apply the term ``valid'' to time-independent point games in the obvious way: a time-independent  point game $(m_1, m_2)$ is valid if $m_1$ is horizontally valid and $m_2$ is vertically valid.  Note that any time-dependent point game $(m_1, \ldots, m_n)$ yields a time-independent point game 
\begin{eqnarray}
\left( m_1 + m_3 + m_5 + \cdots ,  m_2 + m_4 + m_6 + \cdots \right)
\end{eqnarray}
and the latter game is valid if $(m_1, \ldots, m_n)$ is valid.

\begin{remark}
When we use the term \textbf{point game} by itself, we will always mean  a time-dependent point game.
\end{remark}

\subsection{The relationship between point games and coin-flipping protocols}

\label{subsec:relationship}

Now we will state the known results which motivate our study of valid point games.  There is a close correspondence between valid point games and weak coin-flipping protocols, and this correspondence allows us to deduce assertions about coin-flipping protocols (of both existence and impossibility) by studying properties of point games.

\begin{theorem}
\label{thm:wcf2pg}
Suppose that $\mathbf{C}$ is an $n$-round weak coin-flipping protocol with cheating probabilities $P_A^*$ and $P_B^*$, and that $\delta > 0$.  Then, there exists a valid point game $M = (m_1, \ldots, m_n)$ with  initial configuration  $\frac{1}{2} ( \llbracket 1,0 \rrbracket + \llbracket 0,1 \rrbracket )$ and
final configuration $\llbracket P_A^* + \delta, P_B^* + \delta \rrbracket$.
\end{theorem}

If $M$ is a point game from the configuration $\frac{1}{2}(\llbracket 0,1 \rrbracket + \llbracket 1,0 \rrbracket)$ to a single point $[\alpha, \beta]$, then we will naturally refer to the quantity $\max \{ \alpha,\beta \} - 1/2$ as the \textbf{bias} of $M$.  The above theorem can be understood as asserting that if an $n$-round weak coin-flipping protocol exists with bias $\epsilon$, then there are $n$-round valid point games with bias arbitrarily close to $\epsilon$.  
A proof of Theorem~\ref{thm:wcf2pg}, which is based on semidefinite programming duality, is given in \cite{Aharonov:2016}.\footnote{The
proof appears
in sections 1--3 in \cite{Aharonov:2016}.  Note that the definitions of point games used by the authors in \cite{Aharonov:2016} are the ``transpose'' of ours: under their definitions, a valid point game begins with a vertically valid move rather than a horizontally valid move.}
(We note that Theorem~\ref{thm:wcf2pg} also has a converse, although we will not need it here.  See Theorem 3 and Theorem 4 in \cite{Aharonov:2016}.)

\begin{comment}The above theorem asserts that one can always convert a weak coin-flipping protocol into a valid point game if one is willing to accept an arbitrarily small increase in the bias.  The next theorem essentially asserts the reverse.

\begin{theorem}
\label{thm:pg2wcf}
Suppose that $M = (m_1, \ldots, m_n)$ is a valid time-dependent point game whose initial configuration is $(1/2)(\llbracket 0,1 \rrbracket + \llbracket 1,0
\rrbracket)$ and whose outcome is $
\llbracket \alpha, \beta \rrbracket $, and suppose that $\delta > 0$.  Then, there exists a $n$-round weak coin-flipping protocol $\mathbf{C}$ whose cheating probabilities satisfy $P_A^* \leq \alpha + \delta$, and $P_B^* \leq \beta + \delta$.
\end{theorem}

\begin{proof}
This follows from Theorem 3 and Theorem 4
in \cite{Aharonov:2016}. 
\end{proof}

(We note that the proof in \cite{Aharonov:2016} referred to above is nonconstructive, but it has since been made explicit in the more recent paper \cite{Arora:2019}.)

\end{comment}

Next we assert a theorem about the relationship between weak coin-flipping protocols and time-independent point games.

\begin{theorem}
\label{thm:wcf2tpg}
Suppose that $\mathbf{C}$ is an $n$-round weak  coin-flipping protocol with cheating probabilities $\alpha := P_A^*$ and $\beta := P_B^*$, and suppose that $\delta > 0$.  Then, there exists a valid TIPG $R = (r_1, r_2)$ such that
\begin{eqnarray}
r_1 + r_2 & = & -\frac{1}{2}\llbracket 0,1 \rrbracket - \frac{1}{2} \llbracket 1,0 \rrbracket + \llbracket \alpha + \delta, \beta + \delta \rrbracket
\label{ti-target}
\end{eqnarray}
and
\begin{eqnarray}
\left\| r_1 \right\|_1 + \left\| r_2 \right\|_1 & \leq & 2n.
\end{eqnarray}
\end{theorem}

\begin{proof}
By Theorem~\ref{thm:wcf2pg}, there is a valid time-dependent point game $M = (m_1, \ldots, m_n)$ whose sum is equal to the right-hand side of equation (\ref{ti-target}).  Note that since the initial configuration $\frac{1}{2}\llbracket 0,1 \rrbracket + \frac{1}{2} \llbracket 1,0 \rrbracket$ of this game has $1$-norm equal to one, and each move $m_i$ sums to zero, the intermediate configurations all also have $1$-norm equal to one, and therefore $\left\| m_i \right\|_1 \leq 2$ for all $i$.  Therefore the TIPG $(r_1, r_2) := (m_1 + m_3 + m_5 + \cdots , m_2 + m_4 + m_6 + \cdots )$ satisfies the desired conditions.
\end{proof}

We make one final note about symmetry in valid TIPGs.
\begin{remark}
Let us say that a  move $q$ is \textbf{symmetric} if $q^\top = q$, and that a TIPG $Q=(q_1, q_2)$ is symmetric if $q_2 = q_1^\top$.  Note that if a valid TIPG $R = (r_1, r_2)$ achieves a symmetric move (such as the biased coin-flip move $[\frac{1}{2} + \epsilon, \frac{1}{2} + \epsilon] - \frac{1}{2} [0,1] - \frac{1}{2} [1,0]$) then there necessarily exists a symmetric valid TIPG that achieves the same move --- specifically, $R' := ((r_1 + r_2^\top)/2, (r_1^\top + r_2)/2)$.
\end{remark}

\section{The profile of a move}

\label{sec:profile}

We now begin the contributions of this paper.  We start by introducing the idea of a \textbf{profile function} of a move.  If $q$ is a two-dimensional move, then its profile function, denoted $\widehat{q}$, is a real-valued function on $[1, \infty] \times [1, \infty]$.  Profile functions have geometric features that we can use to our advantage when trying to answer questions about point games.

\subsection{Definition}

We begin with the definition of the profile function in the one-dimensional case.  The one-dimensional profile function can be thought of simply as a construction that bundles together the three conditions that define a valid move (Definition~\ref{def:valid}).

\begin{definition}
\label{def:1profile}
For any positive real number $x$, let
$P_x \colon [1, \infty ] \to \mathbb{R}$ be defined by
\begin{eqnarray}
\label{singleprofile}
P_x ( \alpha ) & = & \left\{ \begin{array}{cl} 1 & \hskip0.2in \textnormal{ if } \alpha \in [1,2) \\ \\
(x\alpha  -  x)/(x + \alpha - 2) & \hskip0.2in \textnormal{ if } \alpha \in [2,\infty) \\ \\
x & \hskip0.2in \textnormal{ if } \alpha = \infty .
\end{array}
\right.
\end{eqnarray}
If $x = 0$, then let $P_0 \colon [1, \infty ] \to \mathbb{R}$ be defined by
\begin{eqnarray}
P_0 ( \alpha ) & = & \left\{ \begin{array}{cl} 1 & \hskip0.2in \textnormal{ if } \alpha \in [1,2) \\ \\
0 & \hskip0.2in \textnormal{ if } \alpha \in [2,\infty] .
\end{array}
\right.
\end{eqnarray}
The function $P_x$ is referred to as the \textbf{profile of $x$}.

For any one-dimensional move $\ell \colon \mathbb{R}_{\geq  0} \to \mathbb{R}$, the \textbf{profile of $\ell$}, denoted $\widehat{\ell}$, is the function
from $[1, \infty]$ to $\mathbb{R}$ given by
\begin{eqnarray}
\widehat{\ell} = \sum_x \ell ( x ) P_x.
\end{eqnarray}
\end{definition}
The profile function $\widehat{\ell}$ collects some useful information about $\ell$.  One can easily verify from the definition that
\begin{eqnarray}
\label{prof1}
\widehat{\ell}(1) & = & \sum_x  \ell ( x ) \\
\label{prof2}
\widehat{\ell}(2) & = & \sum_{x>0} \ell ( x ) \\
\label{prof3}
\widehat{\ell}(\infty) & = & \sum_x x \cdot \ell (x).
\end{eqnarray}
For illustration, we give a graph of the function $P_3 ( \alpha )$ (that is, the profile of $\llbracket 3 \rrbracket$) in Figure~\ref{fig:3profile}.
\begin{figure}
    \centering
    \includegraphics[scale=0.5]{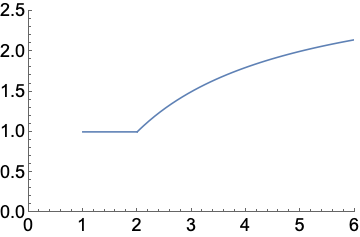}
    \caption{A graph of the function $P_3 ( \alpha )$. (Mathematica)}
    \label{fig:3profile}
\end{figure}

\begin{proposition}
\label{prop:profpos}
A one-dimensional move $\ell$ is valid if and only if $\widehat{\ell}$ is nonnegative and $\widehat{\ell} (1) = 0$.
\end{proposition}

\begin{proof}
Given Definition~\ref{def:valid}, this follows from equations (\ref{prof1}) and (\ref{prof3}) together with the observation that
\begin{eqnarray}
\sum_x \left( \frac{x}{\lambda + x} \right) \ell (x) \geq 0 & \hskip0.2in \Longleftrightarrow \hskip0.2in & \widehat{\ell} \left( \lambda + 2 \right) \geq 0.
\end{eqnarray}
This completes the proof.
\end{proof}

The reader will observe that there are multiple ways that the one-dimensional profile function could have been defined to achieve the property in Proposition \ref{prop:profpos} (e.g., using different ranges for $\alpha$ and different rational expressions). 
This particular choice of definition will make some of the arguments in section~\ref{sec:commbound} mathematically easier.  One of the reasons that this definition is convenient is that the profile of $\llbracket 1 \rrbracket$ is simply the constant function $\alpha \mapsto 1$.  

Next we define the profile of a two-dimensional move.

\begin{definition}
\label{def:2profile}
For any two-dimensional move $q$, the \textbf{profile of $q$}, denoted $\widehat{q}$, is the function
from $[1, \infty ] \times [1, \infty]$ to $\mathbb{R}$ given by
\begin{eqnarray}
\widehat{q}(\alpha, \beta) = \sum_{x,y} q(x,y) P_x(\alpha) P_y ( \beta).
\end{eqnarray}
\end{definition}

As with the one-dimensional profile, some of the values of a two-dimensional profile $\widehat{q}$ have natural expressions --- for example,
\begin{eqnarray}
\widehat{q} ( 1, \infty ) & = & \sum_{x,y} y \cdot q ( x, y )
\end{eqnarray}
and
\begin{eqnarray}
\widehat{q} ( \infty, \infty ) & = & \sum_{x,y} x\cdot y \cdot q ( x, y ).
\end{eqnarray}

The basic motivation to study the two-dimensional profile function is the following proposition.
\begin{proposition}
\label{prop:horizprof}
If $q$ is a move that is either horizontally valid or vertically valid, then $\widehat{q} \geq 0$.
\end{proposition}

\begin{proof}
Suppose that $q$ is horizontally valid.  Then, the single-variable profiles of the rows of $q$ are all nonnegative.  Thus for any
$\alpha, \beta \in [1,\infty]$,
\begin{eqnarray}
\widehat{q}(\alpha, \beta) & = & \sum_y \left( \sum_x q(x,y) P_x(\alpha) \right) P_y ( \beta) \\
& \geq & 0.
\end{eqnarray}
The vertically valid case is similar.
\end{proof}
As a consequence of the above proposition, if $M = (m_1, \ldots, m_n)$ is a valid time-dependent point game and 
$z_0, z_1, \ldots, z_{n-1}, z_n$ are its configurations, we must have
\begin{eqnarray}
\widehat{z_0} ( \alpha, \beta ) \leq \widehat{z_1} ( \alpha, \beta ) \leq \ldots \leq \widehat{z_n} ( \alpha , \beta )
\end{eqnarray}
for any $\alpha, \beta \in [1, \infty]$.
The profile function gives us an infinite family of constraints that must be satisfied by the initial and final configurations of any valid point game.

We make note of some additional elementary facts for later use.
\begin{proposition}
\label{prop:horiz1}
Let $q$ be a horizontally valid move.  If $\alpha \in [1,2)$,
then $\widehat{q} ( \alpha , \beta) = 0$.
\end{proposition}

\begin{proof}
We have
\begin{eqnarray}
\widehat{q} ( \alpha, \beta ) & = &
\sum_y \left[ \left( \sum_{x} q (x, y) \right) P_y ( \beta ) \right] \\
& = & 0,
\end{eqnarray}
as desired.
\end{proof}

\begin{proposition}
\label{prop:horiz2}
Let $q$ be a horizontally valid move.  Then, $\widehat{q} ( \alpha, 1) \geq \widehat{q} ( \alpha , 2 )$ for any $\alpha \in [1, \infty]$.
\end{proposition}

\begin{proof}
By construction, $P_y ( 1 ) \geq P_y ( 2)$ for any $y$.  We have
\begin{eqnarray}
\widehat{q} ( \alpha, 1 ) - \widehat{q} ( \alpha, 2 ) & = &
\sum_y \left[ \left( \sum_{x} q (x, y) P_x ( \alpha ) \right) \left( P_y (1) - P_y ( 2 ) \right) \right].
\end{eqnarray}
Both factors in the summand enclosed by brackets above are nonnegative, and thus the result follows.
\end{proof}

\subsection{The target profiles}

\label{subsec:targ}

From subsection~\ref{subsec:relationship}, we know that if there exists a quantum weak coin-flipping protocol whose bias is less than $\epsilon$ (with $0 < \epsilon < \frac{1}{2}$), then there must exist a valid TIPG $(m_1, m_2)$ such that $m_1 + m_2$ is equal to 
\begin{eqnarray}
v_\epsilon & := & 
\Bigl\llbracket \frac{1}{2} + \epsilon, \frac{1}{2} + \epsilon \Bigr\rrbracket - \frac{1}{2} \llbracket 1,0 \rrbracket
- \frac{1}{2} \llbracket 0,1 \rrbracket.
\end{eqnarray}
We therefore have a crucial interest in the moves $v_\epsilon$.
However, it is mathematically simpler to instead study the family of moves
\begin{eqnarray}
\label{taumove}
t_{\tau} & := & 2 \cdot \llbracket
1,1\rrbracket - \llbracket 2 - \tau, 0 \rrbracket - \llbracket 0, 2 - \tau \rrbracket
\end{eqnarray}
for $0 < \tau < 1$.  Note that if $\tau = 4 \epsilon/(1 + 2 \epsilon)$, then
\begin{eqnarray}
v_\epsilon ( x, y) & = &
\frac{1}{2} t_\tau ( (2-\tau) x, (2-\tau) y),
\end{eqnarray}
and so (using Proposition~\ref{prop:stretch}), valid TIPGs for $v_\epsilon$ correspond exactly to valid TIPGs
for $t_\tau$ via the same linear transformation.

If $R = (r_1, r_2)$ is a valid TIPG such that $r_1 + r_2 = t_\tau$, then we must have $\widehat{r_1} + \widehat{r_2} = \widehat{t_\tau}$.  
The profile function $\widehat{t_\tau}$ can be expressed as follows.
\begin{eqnarray}
\label{tauexplicit}
\widehat{t_\tau} ( \alpha, \beta ) & = & \left\{ \begin{array}{cl} 
2 & \textnormal{ if } \alpha, \beta \geq 2 \\ \\
2 - P_{2-\tau} ( \alpha)  & \textnormal{ if } \alpha \geq 2, \beta < 2 \\ \\
2 - P_{2-\tau} ( \beta) & \textnormal{ if } \alpha < 2, \beta \geq 2 \\ \\
0 & \textnormal{ if } \alpha, \beta < 2 
\end{array}
 \right.
\end{eqnarray}
A graph of an example (with $\tau = 1/10$) is given in Figure~\ref{fig:tenthprofile}.
\begin{figure}
    \centering
    \includegraphics[scale=0.5]{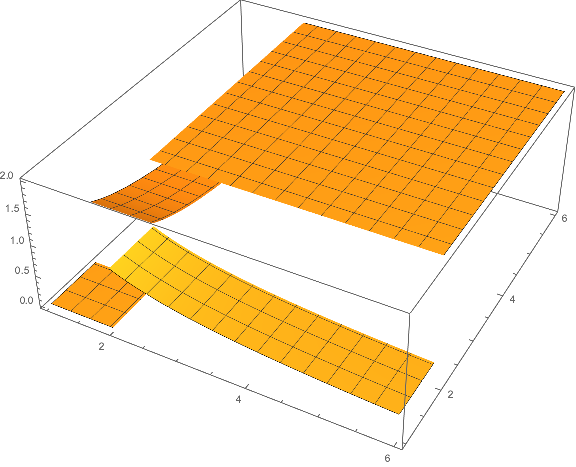}
    \caption{A graph of the function $\widehat{t_{1/10}}$. (Mathematica)}
    \label{fig:tenthprofile}
\end{figure}

\section{Highly concentrated rational functions}

\label{sec:highly}

This section adapts known complex analysis techniques to prove a result that will be needed in section~\ref{sec:commbound}.\footnote{See sections III.1--III.2 in \cite{Nevanlinna:1970} for a more general discussion of some of the techniques used in the current section.}
I am grateful to Alexandre Eremenko for showing me the central method used in this section.

We will be concerned with rational functions on an interval $[a,b]$ that are highly concentrated --- that is, rational functions that are significantly large at some interior point $c \in [ a, b ]$ and are more tightly bounded outside of a small neighborhood of $c$.  Our interest (eventually) will be in studying  deductions that we can make when such a function occurs in the profile of a move.

\subsection{Preliminaries}

We will first reproduce a standard result about the logarithm of the absolute value of an analytic function.
We begin with the following observation: if $f$ is an analytic function on a neighborhood of $\overline{\mathbb{D}}$ which has no zeroes in $\overline{\mathbb{D}}$, then there is a well-defined analytic function $\log f$ on a neighborhood of $\overline{\mathbb{D}}$ such that $2^{\log f} = f$.  We have (see subsection~\ref{subsec:complex}):
\begin{eqnarray}
\frac{1}{2 \pi} \int_{0}^{2 \pi} \log f ( e^{i t} ) dt & = & \log f ( 0 ).
\end{eqnarray}
Since the real part of $\log f(z)$ is precisely $  \log \left|  f(z) \right|$, this proves the following.
\begin{proposition}
\label{prop:concprelim}
Let $f$ be an analytic function on an open neighborhood of $\overline{\mathbb{D}}$ such that $f$ has no zeroes in $\overline{\mathbb{D}}$.  Then,
\begin{eqnarray}
\log \left| f ( 0 ) \right| & = & \frac{1}{2 \pi} \int_{0}^{2 \pi} \log \left| f ( e^{i t} ) \right| dt. \qed
\end{eqnarray}
\end{proposition}

For the result that we will prove in subsection~\ref{subsec:highly}, we will need a similar statement that addresses the case where $f$ is permitted to have zeroes in $\overline{\mathbb{D}}$ and may not be analytic on $\mathbb{S}$.  
This motivates a somewhat more intricate claim.
Let us say that a real-valued function on the unit circle $\mathbb{S}$ is a \textbf{step function} if it is locally constant at all but a finite number of points in $\mathbb{S}$.  
A proof of the following proposition is given in Appendix~\ref{app:logbound}.  
\begin{proposition}
\label{prop:conc}
Let $f$ be a continuous function on $\overline{\mathbb{D}}$ which is analytic on $\mathbb{D}$.  
Suppose that
$b \colon \mathbb{S} \to \mathbb{R}$ is a step function such that
$\log \left| f ( z ) \right|  \leq b ( z )$
for any $z \in \mathbb{S}$. Then,
\begin{eqnarray}
\log \left| f ( 0 ) \right| & \leq & \frac{1}{2 \pi} \int_{0}^{2 \pi} b ( e^{i t} ) dt. \qed
\end{eqnarray}
\end{proposition}

\subsection{The complex values of a highly concentrated rational function}

\label{subsec:highly}

We will now prove the main result of this section.  This result is concerned with rational functions that have real poles --- that is, functions $f \colon \mathbb{C} \cup \left\{ \infty \right\} \to \mathbb{C} \cup \left\{ \infty \right\}$ that can be expressed in the form
\begin{eqnarray}
f(z) & = & \frac{g(z)}{\prod_{i=1}^n (z-c_i)},
\end{eqnarray}
where $g(z)$ is a polynomial and $c_i \in \mathbb{R}$.  We will show that under certain assumptions, these functions must take on large
values on the complex unit circle.

Recall that if $A$ and $B$ are sets and $B \subseteq A$,  then we write $A \smallsetminus B$ for the set of all elements in $A$ that are not in $B$.

\begin{proposition}
\label{prop:complexlarge}
Let $f(z)$ be a rational function whose poles are all real and lie outside of $[-1,1]$.  
Suppose that
\begin{eqnarray}
\left| f ( 0 ) \right| = 1,
\end{eqnarray}
and suppose that  $\delta, \nu \in (0,1)$ are such that 
\begin{eqnarray}
\left| f (z ) \right| \leq \nu && \hskip0.2in \textit{for all } z \in [-1,1] \smallsetminus (-\delta, \delta).
\end{eqnarray}
Then,
\begin{eqnarray}
\label{clconclusion}
\max_{\left| z \right| = 1} \left| f (z ) \right| & \geq & \nu^{-\Omega ( 1/\delta )}.
\end{eqnarray}
\end{proposition}

We illustrate the statement above with an example.
Let
\begin{eqnarray}
\label{econcfunction}
h( z ) & = & \left[ \frac{4(z-1)(z+1)}{(z-2)(z+2)} \right]^{100}.
\end{eqnarray}
This function, which is graphed in Figure~\ref{fig:conc}, satisfies $\left| h ( 0 ) \right| = 1$ and 
\begin{eqnarray}
\left| h (z ) \right| \leq 0.1 && \hskip0.2in \textnormal{for all } z \in [-1,1] \smallsetminus (-0.2,0.2).
\end{eqnarray}
The quantity  $\max_{\left| z \right| = 1 } \left|  h ( z ) \right|$ in this case is indeed quite large (on the order of $10^{20}$).

\begin{figure}
    \centering
    \includegraphics[scale=0.4]{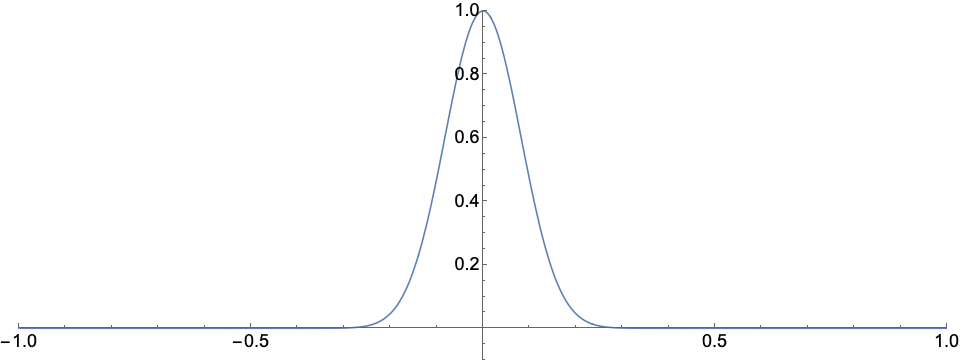}
    \caption{A graph over the interval $[-1,1]$ of the example function in equation~(\ref{econcfunction}). }
    \label{fig:conc}
\end{figure}

\begin{proof}[Proof of Proposition~\ref{prop:complexlarge}]
Our approach is to apply an analytic transformation 
which maps the unit circle $\mathbb{S}$ to the set
\begin{eqnarray}
\mathbb{S} \cup
\left( [-1, 1] \smallsetminus (-\delta, \delta) \right)
\end{eqnarray}
and to thereby reduce the proof to an application of Proposition~\ref{prop:conc} above.  

Let 
$G \colon \overline{\mathbb{D}} \to \overline{\mathbb{H}}$
be defined by
\begin{eqnarray}
\label{gdef}
G ( z ) & = & i \cdot \frac{1 + z}{1-z}.
\end{eqnarray}
Note that this function is a one-to-one mapping.  Its inverse is given by
\begin{eqnarray}
G^{-1} ( z ) & = & \frac{z - i}{z + i}.
\end{eqnarray}
Let $F \colon \overline{\mathbb{H}} \to \overline{\mathbb{H}}$
be the continuous function\footnote{The rational function $z \mapsto  \left( \frac{
z^2  \left| G ( \delta) \right|^2 - 1 }{ \left| G ( \delta) \right|^2 - z^2 } \right)$
on $\overline{\mathbb{H}}$ has two continuous square roots.  We let $F$ be the square root  which maps $\overline{\mathbb{H}}$ into itself.} 
\begin{eqnarray}
\label{fdef}
F (z) & = & \sqrt{ \frac{
z^2  \left| G ( \delta) \right|^2 - 1 }{ \left| G ( \delta) \right|^2 - z^2 }}.
\end{eqnarray}
The function $F$ maps $\infty$ to $G ( \delta)$, maps $0$ to 
$G ( - \delta ) = 1 / G ( \delta)$, and maps $i$ to $i$.  (Diagrams of $F$ and $G$ are given in Appendix~\ref{app:adiag}.)
The image of $\mathbb{R} \cup \{ \infty \}$ under $F$ is the union of $\mathbb{R} \cup \{ \infty \}$ together with the line segment from $G ( \delta )$ to $\infty$, and the line segment from $G ( - \delta )$ to $0$.  

Now let $H$ be  defined by
\begin{eqnarray}
H ( z ) = G^{-1} ( F ( G ( z ) ) ).
\end{eqnarray}
Then,
\begin{eqnarray}
H ( 0 ) & = & 0 \\
H ( 1 ) & = & \delta \\
H ( -1) & = & -\delta.
\end{eqnarray}
The image of the unit circle under $H$ consists of the unit circle, the line segment from $-1$ to $-\delta$, and the line segment from $\delta$ to $1$.
Additionally, if we let $\theta \in [0, \pi/2]$ denote the angle of the unit-length complex number
\begin{eqnarray}
\label{deltafrac}
\frac{i + \delta}{1 + i \delta},
\end{eqnarray}
then the following hold by direct computation:
\begin{eqnarray}
& H \left( e^{i \theta} \right) = 
H \left( e^{- i \theta} \right) = 1 \\
& H \left( - e^{- i \theta} 
\right)
= 
H \left( - e^{i \theta} \right) = -1 .
\end{eqnarray}
The points on the unit circle which lie clockwise between $e^{-i \theta}$ and $e^{i \theta}$ are mapped into the real interval $[\delta, 1]$, and the points which lie clockwise between $-e^{-i \theta}$ and $-e^{i \theta}$ are mapped into the real interval $[-1, -\delta]$.  
All other points on the unit circle remain on the unit circle under the application of $H$.  A diagram of $H$ is given in
Figure~\ref{themaph}.
\begin{figure}
    \centering
\fbox{\includegraphics[scale=0.3]{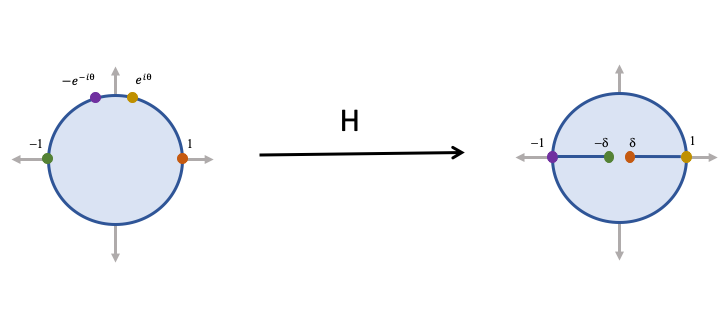}}
    \caption{A diagram of the map $H$.}
\label{themaph}
\end{figure}

We now compute upper bounds on
the function $z \mapsto \left| f ( H (z)) \right|$.  Let $M = \max_{\left| z \right| = 1} \left| f(z) \right|$. Then, by our construction, \begin{eqnarray}
\log \left| f ( H ( e^{i \chi} )) \right| \leq \log M & & \textnormal{ if } \chi \in (\theta, \pi - \theta) \cup (\pi + \theta,  2\pi - \theta), \\
\log \left| f ( H ( e^{i \chi} )) \right| \leq \log \nu & & \textnormal{ if } \chi \in [-\theta, \theta] \cup [\pi - \theta, \pi + \theta]. 
\end{eqnarray}
 Let $T \colon [0, 2 \pi] \to \mathbb{R}$ be defined by 
\begin{eqnarray}
T ( \chi ) & = & \left\{ \begin{array}{ccl} 
\log M & & \textnormal{ if } 
\chi \in (\theta, \pi - \theta) \cup (\pi + \theta,  2\pi - \theta) \\ \\
\log \nu && \textnormal{ otherwise.}
\end{array} \right.
\end{eqnarray}
Then, applying Proposition~\ref{prop:conc},
\begin{eqnarray}
\log \left| f ( 0 ) \right| & = & \log \left| f ( H ( 0 ) ) \right| \\ & \leq & \frac{1}{2 \pi} \int_0^{2 \pi} 
T ( s) d s \\
& = &
\left( 1 - \frac{2 \theta}{\pi} \right) \log M + \left(  \frac{2 \theta}{\pi} \right) \log \nu.
\end{eqnarray}
To complete the proof, it suffices to note that the quantity (\ref{deltafrac}), which was used to define $\theta$, is within distance $O ( \delta )$ from the complex number $i$.  Therefore, $\theta$ itself is within distance  $O ( \delta )$ from $\pi/2$.  Thus we obtain 
\begin{eqnarray}
\log \left| f ( 0 ) \right| & \leq &
O ( \delta ) (\log M )  + (1 - O ( \delta ) ) \log \nu.
\end{eqnarray}
Since $\log \left| f ( 0 ) \right| = 0$ by assumption, we therefore have
\begin{eqnarray}
\frac{- \log \nu}{O ( \delta )} & \leq & \log M,
\end{eqnarray}
which yields the desired result.
\end{proof}

We give a brief discussion to show how Proposition~\ref{prop:complexlarge} can be useful for our purposes.  Suppose that $\ell$ is a one-dimensional move such that
\begin{eqnarray}
\widehat{\ell} ( 4 ) & = & 1
\end{eqnarray}
and that $\delta, \nu \in (0, 1)$ are such that
the inequality $\left| \widehat{\ell} ( \alpha ) \right| \leq \nu$ is always satisfied when \begin{eqnarray}
\alpha \in [3,5] \smallsetminus (4-\delta, 4 + \delta).\end{eqnarray}
Then, for any $\alpha \in [3, 5]$, the rational expression for $\widehat{\ell} ( \alpha )$ is given by
\begin{eqnarray}
\widehat{\ell} ( \alpha ) & = & \sum_x \ell (  x ) \cdot \frac{x \alpha - x}{x + \alpha - 2}.
\end{eqnarray}
By Proposition~\ref{prop:complexlarge}, there exists a unit-length complex number $\zeta$ such that
\begin{eqnarray}
\left| \sum_x \ell ( x ) \cdot \left( \frac{x (4 + \zeta) - x}{x + (4 + \zeta) - 2} \right) \right| & \geq & \nu^{-\Omega ( 1/ \delta )}.
\end{eqnarray}
Therefore, 
\begin{eqnarray}
\nu^{-\Omega ( 1/\delta ) } & \leq & \sum_x \left| \ell ( x ) \right| \cdot \frac{ \left| x (4 + \zeta) - x \right|}{\left| x + (4 + \zeta) - 2 \right|}  \\
& \leq & 
\sum_x \left| \ell ( x ) \right| \cdot \frac{\left| x (4 + 1) - x \right| }{\left| x + (4 - 1) - 2 \right|} \\
& \leq & \sum_x \left| \ell ( x ) \right| \cdot 4.
\end{eqnarray}
Thus we conclude that $\left\| \ell \right\|_1 \geq v^{-\Omega ( 1/ \delta ) }$.  Informally, this means that a one-dimensional move $\ell$ can only achieve a profile that is highly concentrated around $x = 4$ if $\ell$ has exponentially large coefficients.  This is similar to reasoning that we will use in section~\ref{sec:commbound}.  

\section{A lemma on random variables with bounded expectation}

\label{sec:expect}

The following elementary lemma addresses a case where two constraints on the expectation of a random variable approximately determine the value of the random variable.  This lemma will be used (in the context of some artificially constructed random variables) in order to carry out an intermediate step in the main proof of section~\ref{sec:commbound}.

\begin{lemma}
\label{lemma:expectation}
There exists a universal function $\mathbf{A} ( u ) \in O ( \sqrt{ u } )$ such that the following holds.  
If $X$ is any positive real-valued random variable satisfying
\begin{eqnarray}
\mathbb{E} [ X ] & \leq & 1, \\
\mathbb{E} [ 1/ X ] & \leq & 1 + \delta,
\end{eqnarray}
with $\delta > 0$, then
\begin{eqnarray}
\mathbb{P} ( \left| X - 1\right| < \mathbf{A} ( \delta ) ) & \geq & \frac{2}{3}.
\end{eqnarray}
\end{lemma}

\begin{proof} We have
\begin{eqnarray}
\mathbb{E} [ X + 1/X - 2 ] & \leq & \delta,
\end{eqnarray}
and $X + 1/X - 2 \geq 0$.  Therefore,
\begin{eqnarray}
\mathbb{P} [ X + 1/X - 2 \geq 3 \delta  ] & \leq & \frac{1}{3},
\end{eqnarray}
or equivalently,
\begin{eqnarray}
\label{quadevent}
\mathbb{P} [ X^2 - (2 + 3 \delta )X + 1 < 0 ] & \geq & \frac{2}{3}.
\end{eqnarray}
By the quadratic formula, the event on the left side of inequality (\ref{quadevent}) is equivalent to
\begin{eqnarray}
\left| X - (1 + \frac{3}{2} \delta) \right|  & < &  \sqrt{ 3 \delta + \frac{9}{4} \delta^2 }  .
\end{eqnarray}
Thus with probability at least $2/3$, $\left| X - 1 \right|$ is less than
\begin{eqnarray}
\frac{3}{2} \delta + \sqrt{ 3 \delta + \frac{9}{4} \delta^2 }.
\end{eqnarray}
The function above is in $O ( \sqrt{\delta} )$, and this completes the proof.
\end{proof}

\section{The $1$-norm of a time-independent point game}

\label{sec:commbound}

This section will perform most
of the remaining technical work necessary to achieve our main result.
Throughout this section, suppose that $\tau \in (0,1)$, and that $g$ is a horizontally valid move such that $g + g^\top = t_\tau$, where
$t_\tau$ denotes the following move (see subsection~\ref{subsec:targ}):
\begin{eqnarray}
t_\tau & = & 2 \llbracket 1,1 \rrbracket -  \llbracket 2-\tau,0 \rrbracket -  \llbracket 0, 2-\tau \rrbracket.
\end{eqnarray}
In this section we will show that the inequality  $\left\| g \right\|_1 \geq \exp ( \Omega ( \tau^{-1/2} ) )$ must always hold.  This is the result that will be used in section~\ref{sec:lp4} to conclude that any weak coin-flipping protocol that achieves bias $\epsilon$ must involve at least $\exp ( \Omega ( \epsilon^{-1/2} ))$ communication rounds.  

For any $b \geq 0$, let $g_b \colon \mathbb{R}_{\geq 0} \times \mathbb{R}_{\geq 0} \to \mathbb{R}$ denote the move defined by
\begin{eqnarray}
g_b ( x, y) & = & \left\{
\begin{array}{cl} g(x,y) &
\hskip0.2in \textnormal{if } y =b \\ \\
0 & \hskip0.2in \textnormal{if } y \neq b.
\end{array} \right.
\end{eqnarray}
(That is, $g_b$ agrees with $g$ on the horizontal line $y=b$, and is zero elsewhere.)  
Note that since $g$ is horizontally valid, the profile function $\widehat{g_b}$ of $g_b$ is always a nonnegative function.

\subsection{Basic properties of $g$}

\label{subsec:lp1}

By assumption, $g$ is a horizontally valid move and its profile function $\widehat{g}$ satisfies
\begin{eqnarray}
\widehat{g} + \widehat{g}^\top = \widehat{t_\tau}.
\end{eqnarray}  The function $\widehat{t_\tau}$ is written out explicitly in equations (\ref{tauexplicit}) and (\ref{singleprofile}).  The following facts are easily deduced.

\begin{fact}
\label{onequiv}
If $\alpha \in [2, \infty]$, then
$\widehat{g} ( \alpha , \alpha ) = 1$.
\end{fact}

\begin{fact}
\label{gpbound}
If $\alpha, \beta \in [2, \infty]$, then $\widehat{g} ( \alpha, \beta ) \leq 2$.
\end{fact}

By Propositions~\ref{prop:horiz1} and \ref{prop:horiz2}, we have the following.

\begin{fact}
If $\alpha \in [1,2)$ and $\beta \in [1, \infty]$, then $\widehat{g} ( \alpha, \beta ) = 0$.
\end{fact}

\begin{fact}
If $\alpha \in [2,\infty]$, then $\widehat{g} ( \alpha, 1) = 2 - P_{2 - \tau} ( \alpha )$.
\end{fact}

\begin{fact}
\label{fact:inittaubound}
If $\alpha \in [2,\infty]$, then $\widehat{g} ( \alpha, 2) \leq 2 - P_{2 - \tau} ( \alpha )$.
\end{fact}

It is easily seen from Definition~\ref{def:1profile} that
for $\alpha \in [2, \infty]$,
\begin{eqnarray}
P_{2 - \tau} ( \alpha ) & \geq &
P_2 ( \alpha ) - \tau \\
& = & \left( 2 - \frac{2}{\alpha} \right) - \tau.
\end{eqnarray}
Thus from Fact~\ref{fact:inittaubound} we have
\begin{fact}
\label{fact:tbound}
If $\alpha \in [2,\infty]$, then $\widehat{g} ( \alpha, 2) \leq \frac{2}{\alpha} + \tau$.
\end{fact}

\subsection{The isolating function}

\label{subsec:lp2}

\label{subsec:isol}

For any $a \in (2, \infty)$, we know (Fact~\ref{onequiv}) that
\begin{eqnarray}
\sum_b \widehat{g_b} ( a, a) = \widehat{g} ( a, a ) = 1,
\label{majcont}
\end{eqnarray}
and that each term $\widehat{g_b} ( a, a )$ in the above summation is nonnegative.  Our next goal is to show that the majority of the contribution to $\widehat{g} ( a, a )$ comes from the terms $\widehat{g_b} ( a, a )$ for which $b$ is close to $a$.  This is formally stated as follows.

\begin{proposition}
\label{prop:isolating}
There is a universal function $\mathbf{I} (u ) \in O ( \sqrt{u} )$ such that for any $a \in [3,5]$,
\begin{eqnarray}
\sum_{b \colon \left| b - a \right| < \mathbf{I} ( \tau ) } \widehat{g_b} ( a, a ) & \geq & \frac{2}{3}.
\end{eqnarray}
\end{proposition}

We note that the choice of the interval $[3,5]$ in this statement is somewhat arbitrary --- the proof method that follows can be used to prove the same statement over any interval $[p,q]$ for which $2 < p < q < \infty$ (with a different choice of function $\mathbf{I}$).  The reason for making this type of restriction is that it facilitates calculations involving universal big-$O$ error terms.

\begin{proof}[Proof of Proposition~\ref{prop:isolating}]

Let $Y$ be the stochastic map from $[3,5]$ to $\mathbb{R}_{> 0}$ defined by
\begin{eqnarray}
\mathbb{P} ( Y (a ) = b ) & = &  \widehat{g_b} ( a, a ).
\end{eqnarray}
Note that by construction, for any $a \in [3,5]$ and any $b >0$, the function on $[1, \infty]$ defined by
\begin{eqnarray}
\beta & \mapsto & \widehat{g_b} ( a, \beta )
\end{eqnarray}
is a scalar multiple of the profile function of $b$ (that is, $\beta \mapsto P_b ( \beta )$).  Therefore,
\begin{eqnarray}
\widehat{g_b} ( a, 2) = 
P_b ( 2 ) \cdot \frac{\widehat{g}_b ( a, a ) }{P_b ( a )} = 1 \cdot \frac{\mathbb{P} ( Y(a ) = b) }{P_b ( a )} \\
\widehat{g_b} ( a, \infty) =
P_b ( \infty ) \cdot \frac{\widehat{g}_b ( a, a )}{P_b ( a )} = 
b \cdot \frac{\mathbb{P} ( Y(a ) = b )}{P_b ( a )}
\end{eqnarray}
which implies
\begin{eqnarray}
\widehat{g} ( a, 2 ) & = & 
\mathbb{E} \left[ \frac{1}{P_{Y(a)}(a)} \right] \\
\widehat{g} ( a, \infty ) & = & 
\mathbb{E} \left[ \frac{Y(a) }{P_{Y(a)}(a)} \right],
\end{eqnarray}
where the expectations are taken over the random variable $Y(a)$.
Expanding these expressions using the definition of the profile function, we have
\begin{eqnarray}
\widehat{g} ( a, 2 ) & = & 
\mathbb{E} \left[ \frac{Y(a) + a - 2}{Y(a) \cdot a - Y(a) } \right] \\
\label{2expformula}
& = & \mathbb{E} \left[ \left( \frac{1}{a-1} \right) + \left( \frac{1}{Y(a)} \right) \left( \frac{a-2}{a-1} \right) \right]
\end{eqnarray}
and
\begin{eqnarray}
\label{infexpformula}
\widehat{g} ( a, \infty ) 
& = & \mathbb{E} \left[  Y(a) \left( \frac{1}{a-1} \right) +   \left( \frac{a-2}{a-1} \right)  \right].
\end{eqnarray}
By Fact~\ref{fact:tbound} and Fact~\ref{gpbound},
\begin{eqnarray}
\widehat{g} ( a, 2 ) & \leq &
\frac{2}{a} + \tau, \\
\widehat{g} ( a, \infty) & \leq & 2.
\end{eqnarray}
Combining the above two inequalities with formulas (\ref{2expformula}) and (\ref{infexpformula}) above yields
\begin{eqnarray}
\label{preconeover}
\mathbb{E} \left[ \left( \frac{1}{Y(a)} \right) \left( \frac{a-2}{a-1} \right) \right] & \leq & \frac{a-2}{a (a-1)} + \tau \\
\mathbb{E} \left[ Y ( a ) \left( \frac{1}{a-1} \right) \right] & \leq & \frac{a}{a-1}.
\end{eqnarray}
Multiplying the equations above by $[a(a-1)/(a-2)]$ and $[(a-1)/a]$ respectively, we obtain
\begin{eqnarray}
\label{oneover}
\label{seventau}
\mathbb{E} \left[  \frac{a}{Y(a)}  \right] & \leq & 1 + 7 \tau \\
\mathbb{E} \left[ \frac{ Y ( a )}{a} \right]  & \leq & 1,
\end{eqnarray}
where, in (\ref{seventau}), we used the fact that the function $a \mapsto [a(a-1)/(a-2)]$ on the interval $[3,5]$ does not exceed $7$.

If we let $\mathbf{A} ( u )$ be the function from Lemma~\ref{lemma:expectation}, then
\begin{eqnarray}
\mathbb{P} \left(  \left| \frac{Y ( a ) }{a} - 1 \right| < \mathbf{A} ( 7 \tau ) \right) & \geq & \frac{2}{3}.
\end{eqnarray}
Therefore,
\begin{eqnarray}
\mathbb{P} \left(  \left| Y ( a ) - a \right| < a \mathbf{A} ( 7 \tau ) \right) & \geq & \frac{2}{3}.
\end{eqnarray}
Letting $\mathbf{I} ( u ) := 5 \mathbf{A} ( 7 u)$ therefore yields 
\begin{eqnarray}
\mathbb{P} \left(  \left| Y ( a ) - a \right| < \mathbf{I} ( \tau ) \right) & \geq & \frac{2}{3}
\end{eqnarray}
for any $a \in [3,5]$.
By the definition of the stochastic map $Y$, this implies the desired result.
\end{proof}

We refer to $\mathbf{I}$ as the \textbf{isolating function}.  

\subsection{A lower bound on 
$\left\| g \right\|_1$}

\label{subsec:lp3}

We are now ready to prove a lower bound on $\left\| g \right\|_1$ in terms of $\tau$.  We accomplish this by studying the behavior
of the part of the move $g$ that is concentrated near the horizontal line $y = 4$.  Precisely, we will be concerned with the move
\begin{eqnarray}
\textbf{g} & := & \sum_{b \colon \left| b - 4 \right| < \mathbf{I} ( \tau ) } g_b, 
\end{eqnarray}
where $\mathbf{I}$ is the isolating function from subsection~\ref{subsec:isol}.  (See Figure~\ref{fig:gregion}.)

\begin{proposition}
\label{prop:isolating2}
For all $a \in [3,5]$ such that $\left| a - 4 \right| \geq 2 \mathbf{I} ( \tau )$, we must have 
\begin{eqnarray}
\widehat{\mathbf{g}} (a,a ) & \leq & \frac{1}{3}. \end{eqnarray}
\end{proposition}

\begin{proof}
For any such $a$, the move
\begin{eqnarray}
\label{offcenter}
\sum_{b \colon \left| b-a \right|< \mathbf{I} ( \tau ) } g_b, 
\end{eqnarray}
has disjoint support from that of $\textbf{g}$.  The sum of the profile of  (\ref{offcenter}) and the profile of $\textbf{g}$ is therefore upper bounded by the profile of $g$.  By Proposition~\ref{prop:isolating} and Fact~\ref{onequiv}, we have
$\widehat{\textbf{g}} 
( a, a )  \leq  \frac{1}{3}$,
as desired.
\end{proof}

\begin{figure}
\begin{center}
\includegraphics[scale=0.6]{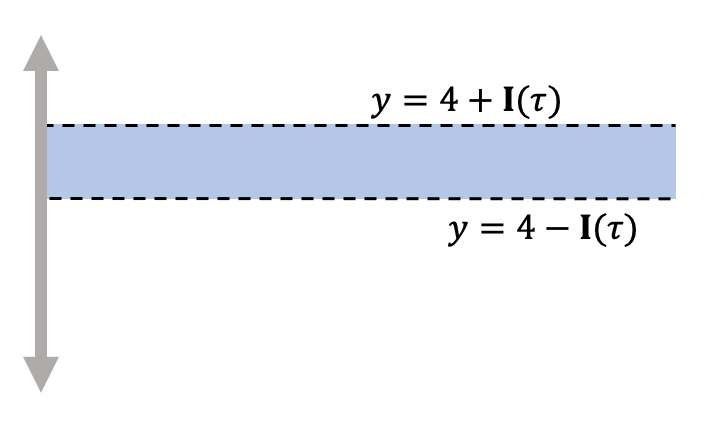}
\end{center}
\caption{The move $\mathbf{g}$ is supported within the shaded blue region.}
\label{fig:gregion}
\end{figure}

\begin{proposition}
\label{prop:1normbound}
The $1$-norm $\left\| g \right\|_1$ of $g$ satisfies
\begin{eqnarray}
\left\| g \right\|_1 & \geq & \exp ( \Omega ( \tau^{-1/2} )).
\end{eqnarray}
\end{proposition}

\begin{proof}
Consider the function on $(2, \infty)$ given by
\begin{eqnarray}
z & \mapsto & \widehat{\textbf{g}} ( z, z ).
\end{eqnarray}
This function can be written out explicitly as follows:
\begin{eqnarray}
\label{thediagrational}
z & \mapsto & \sum_{x,y} \mathbf{g} ( x, y ) \cdot \frac{ xz - x}{x + z - 2 } 
\cdot \frac{ yz - y}{y + z - 2 } 
\end{eqnarray}
Let $D \colon \mathbb{C} \cup \left\{ \infty \right\} \to \mathbb{C} \cup \left\{ \infty \right\}$ denote the complex-valued rational function given by expression (\ref{thediagrational}).
By Propositions~\ref{prop:isolating} and \ref{prop:isolating2},
the function $D$ satisfies $D(4) \geq 2/3$, and $0 \leq D ( z ) \leq 1/3$ for all $z$ in the set
\begin{eqnarray}
[3,5] \smallsetminus (4- 2 \mathbf{I} ( \tau ) , 4 +2 \mathbf{I} ( \tau )).
\end{eqnarray} 
Applying Proposition~\ref{prop:complexlarge} (with appropriate affine transformations), we find that there exists a unit-length complex number $\zeta$ such that
\begin{eqnarray}
\left| D ( 4 +  \zeta) \right| & \geq & \exp ( \Omega ( 1/\mathbf{I} ( \tau ) ) \\
& \geq & \exp ( \Omega ( \tau^{-1/2} )). \label{penultimate}
\end{eqnarray}

Next we obtain an upper bound on $\left| D ( 4 + \zeta ) \right|$ in terms of $\left\| g \right\|_1$.  We have 
\begin{eqnarray}
\left| D ( 4 + \zeta ) \right| & \leq & \sum_{x,y} \left| \mathbf{g} ( x, y ) \right| \cdot \left| \frac{ x(4 + \zeta) - x}{x + (4 + \zeta) - 2 } \right|
\cdot \left| \frac{ y (4 + \zeta) - y}{y + (4 + \zeta) - 2 } \right|.
\end{eqnarray}
It is easy to see that the second and third factors in the summand above are each no more than $4$.  Therefore,
\begin{eqnarray}
\left| D ( 4 + \zeta ) \right| & \leq & 16 \left\| \mathbf{g } \right\|_1 \\
& \leq & 16 \left\| g \right\|_1.
\end{eqnarray}
Combining the above bound with inequality (\ref{penultimate}) above yields the desired result.
\end{proof}

\section{Main result}

\label{sec:lp4}

We can now tie together the results of section~\ref{sec:coin} and section~\ref{sec:commbound} to achieve our main result.

\begin{proposition}
\label{prop:epssec}
Let $\epsilon \in (0, \frac{1}{2})$, and suppose that $M = (m_1, m_2)$ is a valid time-independent point game such that
\begin{eqnarray}
m_1 + m_2 & = & \Bigl\llbracket \frac{1}{2} + \epsilon, \frac{1}{2} + \epsilon \Bigr\rrbracket - \frac{1}{2} \cdot \llbracket 1,0 \rrbracket - \frac{1}{2} \cdot \llbracket 0,1 \rrbracket.
\end{eqnarray}
Then,
\begin{eqnarray}
\left\| m_1 \right\|_1 + \left\| m_2 \right\|_1 & \geq & \exp ( \Omega ( \epsilon^{-1/2} )).
\end{eqnarray}
\end{proposition}

\begin{proof}
Define moves $m'_1, m'_2$ by
\begin{eqnarray}
m'_i(x,y) & = & 2 m_i \left( x \left( \frac{1}{2} + \epsilon \right), y
\left( \frac{1}{2} + \epsilon \right) \right).
\end{eqnarray}
Then, $(m'_1, m'_2)$ is a valid TIPG which achieves the move $2 \llbracket 1, 1 \rrbracket - \llbracket 2-\tau, 0 \rrbracket - \llbracket 0, 2 - \tau \rrbracket$, where $\tau = 4 \epsilon/(1+2 \epsilon)$.  Moreover, if we let $m' = [m'_1+(m'_2)^\top]/2$, then $(m', (m')^\top)$ is a symmetric valid TIPG that achieves the same move.  By Proposition~\ref{prop:1normbound},
\begin{eqnarray}
\left\| m' \right\|_1 & \geq & \exp ( \Omega ( \tau^{-1/2} )).
\end{eqnarray}
It is obvious that $\left\| m_1 \right\|_1 + \left\| m_2 \right\|_1 \geq \left\| m' \right\|_1$ and that $\tau \leq O ( \epsilon)$.  This completes the proof. 
\end{proof}

\begin{theorem}
\label{thm:final}
Suppose that $\mathbf{C}$ is an $n$-round weak coin-flipping protocol with cheating probabilities $P_A^* \leq \frac{1}{2} + \epsilon$ and $P_B^* \leq \frac{1}{2} + \epsilon$.  Then, 
\begin{eqnarray}
\label{thefinalinequality}
n & \geq & \exp ( \Omega ( \epsilon^{-1/2} )).
\end{eqnarray}
\end{theorem}

\begin{proof}
Combining Proposition~\ref{prop:epssec} with  Theorem~\ref{thm:wcf2tpg}, we find that for any $\delta > 0$,
\begin{eqnarray}
n & \geq & \exp ( \Omega ((\epsilon + \delta)^{-1/2 })).
\end{eqnarray}
Since the above inequality is true for any positive real number $\delta$, the desired result follows.
\end{proof}

%\section{Numerical results}

\section{Further directions}

\label{sec:dir}

A natural next step would be to compute an explicit function which would serve as a lower bound for $n$ in  Theorem~\ref{thm:final}.  This is a matter of tracing through the steps of the proof, and should not be difficult.  Explicit bounds on $n$ will open the door to searching for quantum weak coin-flipping protocols that are optimized
for the number of communication rounds (at a particular bias $\epsilon$). 

One can also try to lower bound the amount of quantum memory needed to achieve weak coin-flipping for a given bias.
As discussed in~\cite{Aharonov:2016}, the quantum memory used by a protocol is related to the size of the support of its point games.  Some of the same techniques used in this paper might be applicable to proving lower bounds on quantum memory size.

Can a related impossibility result be proved for strong quantum coin-flipping?  A.~Kitaev showed that any strong coin-flipping protocol must have bias at least $\frac{\sqrt{2}}{2} - \frac{1}{2} \approx 0.207$.  Meanwhile, Chailloux and Kerenidis~\cite{Chailloux:2017} proved, by building on Mochon's work on quantum weak coin flipping with vanishing bias~\cite{Mochon:2007}, that strong coin-flipping is possible with bias arbitrarily close to $\frac{\sqrt{2}}{2} - \frac{1}{2}$.  One could try to prove that strong coin-flipping with bias approaching $\frac{\sqrt{2}}{2} - \frac{1}{2}$ requires a large amount of communication.

Lastly, I will note that although we have found that the moves (\ref{taumove}) that define quantum coin-flipping are exponentially hard to achieve by valid point games, my experience so far suggests this is a uniquely difficult family of moves.   It may be worth exploring whether there are other simple classes of moves that can be more easily achieved by valid point games, and exploring whether such classes could have applications to positive results in two-party cryptography.

\bibliographystyle{plain}

\bibliography{FullPaper}

\appendix

\section{Appendices}

\subsection{Proof of Proposition~\ref{prop:conc}}

\label{app:logbound}

We state some elementary facts.
\begin{fact}
\label{fact:zeroes}
If $f$ is an analytic function on an open neighborhood of a compact set $R \subseteq \mathbb{C}$, and $f$ is not the zero function, then $f$ has only a finite number of zeroes in $R$.
\end{fact}

\begin{fact}
\label{fact:unif}
If $f \colon \overline{\mathbb{D}} \to \mathbb{C}$ is a continuous function, then the family of functions on $\mathbb{S}$ given by
\begin{eqnarray}
z \mapsto f ( cz ),
\end{eqnarray}
for $c \in [0,1]$, converges uniformly to $f_{\mid \mathbb{S}}$ as $c \to 1$.
\end{fact}

We next prove modified versions of Proposition~\ref{prop:concprelim}.

\begin{proposition}
\label{prop:concprelim2}
Let $f$ be an analytic function on an open neighborhood of $\overline{\mathbb{D}}$ such that $f$ has no zeroes on the unit circle $\mathbb{S}$.  Then,
\begin{eqnarray}
\log \left| f ( 0 ) \right| & \leq & \frac{1}{2 \pi} \int_{0}^{2 \pi} \log \left| f ( e^{i t} ) \right| dt.
\end{eqnarray}
\end{proposition}

\begin{proof}
The function $f ( z )$ can be written as $f (z ) = f_1 ( z ) f_2 ( z )$, where $f_1$ is a polynomial with roots in $\mathbb{D}$, and $f_2$ has no zeroes on $\overline{\mathbb{D}}$.  By Proposition~\ref{prop:concprelim},
\begin{eqnarray}
\label{f2eq}
\log \left| f_2 ( 0 ) \right| & = & \frac{1}{2 \pi} \int_{0}^{2 \pi} \log \left| f_2 ( e^{i t} ) \right| dt,
\end{eqnarray}
while direct computation shows that
\begin{eqnarray}
\label{f1ineq}
\log \left| f_1 ( 0 ) \right| & \leq & \frac{1}{2 \pi} \int_{0}^{2 \pi} \log \left| f_1 ( e^{i t} ) \right| dt.
\end{eqnarray}
Summing equation (\ref{f2eq}) and inequality (\ref{f1ineq}) yields the desired result.
\end{proof}

By affine transformation, we have the following corollary.

\begin{corollary}
\label{cor:concprelim}
Let $f$ be an analytic function on an open neighborhood of a closed disc $\overline{\mathbb{D}}(z,r)$ such that $f$ has no zeroes on $\mathbb{S}(z,r)$.  Then,
\begin{eqnarray}
\log \left| f ( z ) \right| & \leq & \frac{1}{2 \pi} \int_{0}^{2 \pi} \log \left| f ( z + r e^{i t} ) \right| dt. \qed
\end{eqnarray}
\end{corollary}

Now we will prove the desired result.

\textit{Proof of Proposition~\ref{prop:conc}.}
If $f$ is the zero function, then Proposition~\ref{prop:conc} is trivial, so we will assume that $f$ is not identically zero.
Take any $\delta > 0$.  By Fact~\ref{fact:unif} above, we can find $\epsilon > 0$ such that the following inequality holds on the annulus $\overline{\mathbb{D}} \smallsetminus \mathbb{D} ( 0, 1 - \epsilon)$:
\begin{eqnarray}
\left| f(z ) - f \left( \frac{z}{ \left| z \right| } \right) \right| & \leq & \delta.
\end{eqnarray}
Choose (using Fact~\ref{fact:zeroes}) a real number $\zeta \in [\epsilon/2,\epsilon]$ such that $f$ has no zeroes on $\mathbb{S} ( 0, 1-\zeta )$. Then, applying Corollary~\ref{cor:concprelim},
\begin{eqnarray}
\log \left| f ( 0 ) \right| & \leq &
\frac{1}{2 \pi} \int_0^{2 \pi} 
\log \left| f ( (1 - \zeta) e^{it} ) \right| dt \\
& \leq &  \frac{1}{2 \pi} \int_0^{2 \pi} \log \left( \left| f ( e^{it} ) \right| + \delta  \right) dt \\
\label{bupperbound}
& \leq &  \frac{1}{2 \pi} \int_0^{2 \pi} \log \left( 2^{b (  e^{it} )} + \delta  \right) dt.
\end{eqnarray}

Since the upper bound (\ref{bupperbound}) holds for any $\delta > 0$, the desired claim follows.  $\qed$

\subsection{Additional diagrams for subsection~\ref{subsec:highly}}
\label{app:adiag}

Below, diagrams are given for the map $G$ (\ref{gdef}) and the map $F$ (\ref{fdef}).
\vskip0.2in
\begin{center}
\fbox{\includegraphics[scale=0.3]{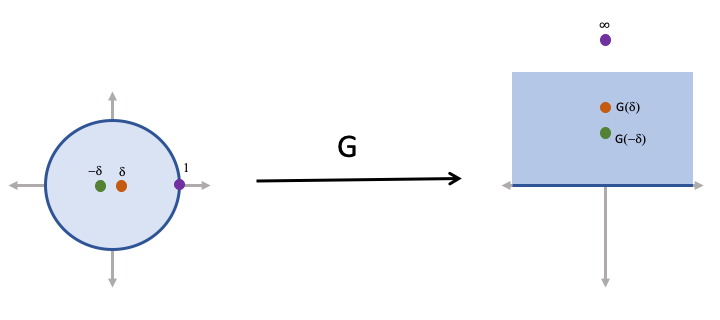}}
\end{center}
\vskip0.2in
\begin{center}
\fbox{\includegraphics[scale=0.3]{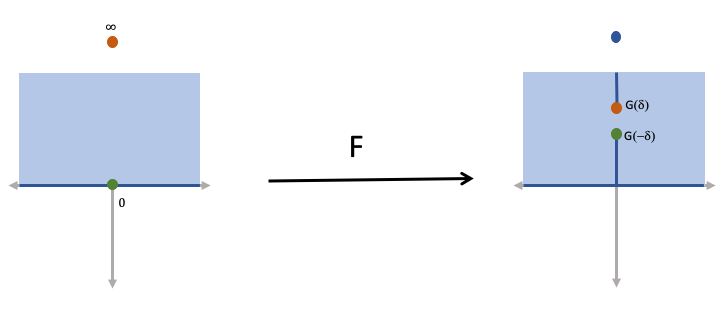}}
\end{center}

\end{document}